\documentclass[11pt]{article}

\usepackage{lineno}
\usepackage{amssymb,amsthm,amsmath}
\usepackage{graphicx}
\usepackage[small]{caption}
\usepackage{subcaption}
\usepackage{epsfig}
\usepackage{amsfonts}
\usepackage{pdfpages}

\usepackage{psfrag}
\usepackage{epsfig}
\usepackage{epsf}
\usepackage{latexsym}
\usepackage{enumerate}
\usepackage{bbm}
\usepackage{complexity}
\usepackage{url}
\usepackage{hyperref}

\newtheorem{theorem}{Theorem}
\newtheorem{lemma}{Lemma}

\newtheorem{corollary}{Corollary}

\newcommand{\etal}{{et~al.}}
\newcommand{\ie}{{i.e.}}
\newcommand{\eg}{{e.g.}}

\newcommand{\opt}{\textsf{OPT}}

\newcommand{\conv}{{\rm conv}}

\def\Ext{{\tt Ext}}

\newcommand{\RR}{\mathbb{R}} %  set of real numbers
\newcommand{\eps}{\varepsilon}

\def\Prob{{\rm Prob}}

\newcommand{\e}{\mathrm{e}}

\newcommand{\later}[1]{}
\newcommand{\old}[1]{}

\title{Two-sided convexity testing with certificates
  \footnote{A preliminary version appears in 
Proceedings of the 12th Japanese-Hungarian Symposium
  on Discrete Mathematics and Its Applications, March 2023, Budapest, Hungary.
  The current expanded version corrects some inaccuracies present there.}}
  
\author{
Adrian Dumitrescu\thanks{%
Algoresearch L.L.C., Milwaukee, WI, USA. 
Email~\texttt{ad.dumitrescu@algoresearch.org}.}}

\begin{document}

\maketitle

\begin{abstract}
  We revisit the problem of property testing for convex position for point sets in $\RR^d$.
  Our results draw from previous ideas of Czumaj, Sohler, and Ziegler (ESA 2000).
  First, the algorithm is redesigned and its analysis is revised for correctness.
  Second, its functionality is expanded by (i)~exhibiting both negative and positive certificates
  along with the convexity determination, and (ii)~significantly extending the input range
  for moderate and higher dimensions. 
  
  The behavior of the randomized tester is as follows: (i)~if $P$ is in convex position, it accepts;
  (ii)~if $P$ is far from convex position, with probability at least $2/3$, it rejects and outputs a
  $(d+2)$-point witness of non-convexity as a negative certificate; (iii)~if $P$ is close to convex position,
  with probability at least $2/3$, it accepts and outputs an approximation of the largest subset in convex position.
  The algorithm examines a sublinear number of points and runs in subquadratic time for every
  fixed dimension $d$.

\end{abstract}

\section{Introduction} \label{sec:intro}

A set of points in the $d$-dimensional space $\RR^d$ is said to be:
(i)~in \emph{general position} if any at most $d+1$ points are \emph{affinely independent}; and
(ii)~in \emph{convex position} if none of the points lies in the convex hull of the other points.
It is known that every set of $n$ points in general position in the plane contains $(1-o(1))\log{n}$
points in convex position, and this bound is tight up to lower-order terms~\cite{ES35,Suk17}. 
For $d\geq 3$, by the Erd\H{o}s--Szekeres theorem, every set of $n$ points in general position
in $\mathbb{R}^d$ contains $\Omega(\log n)$ points in convex position: it suffices to find points
whose projections onto a generic plane are in convex position.
On the other hand, for every fixed $d\geq 2$,
K{\'a}rolyi and Valtr~\cite{KV03} and Valtr~\cite{Va92} constructed  $n$-element
sets in general position in $\mathbb{R}^d$ in which no more than $O(\log^{d-1}n)$ points 
are in convex position. A recent result of Pohoata and Zakharov~\cite{PZ22} shows
that a set of $n$ points in $\RR^d$, $d \geq 3$, already contains a subset of $\omega(\log{n})$
points in convex position. 

Given a point set in general position in $\RR^d$, the problem of computing a maximum-size
subset in convex position can be solved in polynomial time for $d=2$ by the dynamic programming
algorithm of Chv\'atal and Klincsek~\cite{CK80}; their algorithm runs in $O(n^3)$ time.
In contrast, the general problem in $\RR^d$ was shown to be $\NP$-complete  for every $d \geq 3$
by Giannopoulos, Knauer, and Werner~\cite{GKW13}, and moreover, no approximation algorithm is known. 

Throughout this paper we assume (in a standard fashion) that the input set is in \emph{general position}.  
For Theorems~\ref{thm:chazelle} and~\ref{thm:chan} and Corollary~\ref{cor:chan},
let $P$ be a set of $n$ points in $\RR^d$, where $d$ is considered constant.

The complexity of computing the convex hull of $n$ points in $\RR^d$ is summarized in the following 
result of Chazelle; see also~\cite{BCK+08,Sei17}. 

\begin{theorem} {\rm (Chazelle~\cite{Chaz93})} \label{thm:chazelle}
Given $P$, the convex hull of $P$ can be computed in $O(n \log{n} + n^{\lfloor d/2 \rfloor})$ time
  using $O(n^{\lfloor d/2 \rfloor})$ space, which is asymptotically worst-case optimal.
\end{theorem}

It is known that the number of faces, $f$, of the output polytope is $\Theta(n^{\lfloor d/2 \rfloor})$
in the worst case~\cite{McM70}, \ie, exponential in $d$. On the other hand, a result of Chan shows that
the set of extreme points of a set of $n$ points in $\RR^d$ can be computed in subquadratic time
and essentially faster when their number $h$ is small.

\begin{theorem} {\rm (Chan~\cite{Chan96})} \label{thm:chan}
  Given $P$, the $h$ extreme points of $P$ can be computed in time \linebreak
\begin{equation} \label{eq:chan}
  T(n,h) = O\left(n \log^{O(1)}{h} + (nh)^{\frac{\lfloor d/2 \rfloor}{\lfloor d/2 \rfloor +1}} \log^{O(1)}{n}\right).
\end{equation}
\end{theorem}

Taking $n=h$ in the above expression yields a time that suffices for testing whether a set of $n$ points
is in convex position. From the other direction, it is conjectured that the problem of testing
whether a set $P$ is in convex position is asymptotically as hard as the problem of computing
all extreme points of $P$~\cite{CSZ-esa00}.

\begin{corollary} {\rm (Chan~\cite{Chan96})} \label{cor:chan}
  Given $P$, determining whether $P$ is in convex position can be done in time
  $T(n,n) = O\left(n^{\frac{2\lfloor d/2 \rfloor}{\lfloor d/2 \rfloor +1}} \log^{O(1)}{n}\right)$.
\end{corollary}

For instance, the running time in Corollary~\ref{cor:chan} is
$O\left(n \log^{O(1)}{n} \right)$ for $d=2,3$,
$O\left(n^{4/3} \log^{O(1)}{n} \right)$ for $d=4,5$, 
$O\left(n^{3/2} \log^{O(1)}{n} \right)$ for $d=6,7$,
%$O\left(n^{8/5} \log^{O(1)}{n} \right)$ for $d=8,9$
and subquadratic in any fixed dimension $d$.  

Similarly, the following holds (see Corollary 3.4 in~\cite{Chan96}).

\begin{corollary} {\rm (Chan~\cite{Chan96})} \label{cor:chan2}
  Given $P$ and $S \subset P$, where $s=|S|$,
  determining whether all points in $S$ are extreme in $P$ can be done in time
  $T(n,s) = O\left(n \log^{O(1)}{s} + (ns)^{\frac{\lfloor d/2 \rfloor}{\lfloor d/2 \rfloor +1}} \log^{O(1)}{n}\right)$.
\end{corollary}

In \emph{property testing} one is concerned with the design of faster algorithms for approximate
decision making~\cite{Gol17}. In this scenario, instead of determining whether an input has a specific property,
one determines if the input is \emph{far} or perhaps \emph{close} from satisfying that property.
Such approximate decisions, usually involving random sampling or shortcuts in the computation,
may be valuable in settings in which an exact decision is infeasible or just more expensive.
For example, one may be interested in determining, given an input point set, how far it stands from being in
convex position without needing to spend all resources that would be required for computing the convex hull
of the respective set. Such a tool is obviously useful in the general area of testing properties of geometric objects
and visual images for distinguishing a convex shape among other shapes. 

The goal of \emph{property testing} is to develop efficient \emph{property testers}.
Ideally, such a tester makes a sublinear number of queries of the input set, \ie, it does not look at the entire input. 
However, this does not mean --- even for the ideal case --- that the tester runs in time that is sublinear in the size
of the input; in fact, it often doesn't.  
Moreover, if the tester is  also required to return a possibly large subset of the input set (depending on the outcome)
as a certificate, then its time requirements may be further increased. 

Here we focus on the testing of \emph{convex position}.
As in the context of randomized algorithms, approximately deciding means returning the correct answer
with some confidence, specifically with probability at least $2/3$ as described below, see, \eg, \cite{MU17};
however, the $2/3$ threshold is not set in stone. 
%For instance, in regard to the previous point on running time, it is worth noting that already for the plane ($d=2$),
%testing for convex position by the algorithm in~\cite{CSZ-esa00} takes $O(n^{2/3} \eps^{-1/3} \log{(n /\eps)})$,
%which is $\Theta(n \log{n})$ if $\eps =\Theta(1/n)$; running times for larger $d$ are even higher. 

Testing algorithms may use samples of different sizes. Some intuition is as follows.
Suppose that the input is far from convex position; the algorithm is
likely to reject on large samples (the larger the sample, the easier it will be to find that out),
and is likely to accept on small samples (the smaller the sample, the easier the algorithm will be fooled).
On the other hand, if the input is close to convex position,
the smaller the sample, the easier it will be for the algorithm to accept.

A key distinction with regard to the action (accept or reject) is that closeness must fit the goal,
\ie, far and close need to be quantified appropriately.
As it turns out, rejecting an input that is far from convex position is relatively insensitive to the distance from
convex position. However, when accepting an input that is close to convex position, the input must be really close.

\subsection{Preliminaries} \label{subsec:prelim}

\paragraph{Definitions and notation.}
Let $0 <\eps< 1/2$. A set $P$ of $n$ points is $\eps$-\emph{far} from convex position
if there is no set $X \subset P$ of size at most $\eps n$ such that $P \setminus X$ is in convex position. 
Otherwise, \ie, if there is a set $X \subset P$ of size at most $\eps n$ such that $P \setminus X$ is in convex position,
$P$ is $\eps$-\emph{close} to convex position. See Fig.~\ref{fig:examples}.
For a set point $P$, let $\Ext(P)$ denote the set of extreme points of $P$.

\begin{figure}[htbp]
 \centering \includegraphics[scale=0.8]{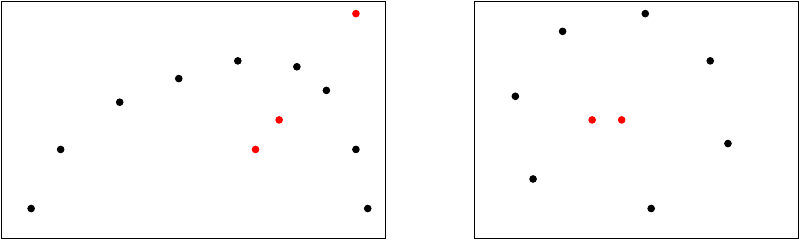}
 \caption{A $12$-point set that is $1/4$-close to convex position (left), and
 a $9$-point set that is $2/9$-close to convex position (right). 
 Both sets are $1/5$-far from convex position.}
	\label{fig:examples}
\end{figure}

Here we use the convention that the approximation ratio of an algorithm is smaller than $1$ for a maximization
problem and larger than $1$ for a minimization problem; see~\cite{WS11}.
Unless specified otherwise, all logarithms are in base $2$.
For a set $W \subset \RR^d$, its \emph{interior} is denoted by $\mathring{W}$.

%Let $[n]$ denote the set $\{1,2,\ldots,n\}$.

\paragraph{Nonconvexity certificates.}

By the well-known Carath\'eodory's Theorem, see, \eg, \cite[p.~6]{Mat02},
if $X$ is finite point set in $\RR^d$, every point of $X$ can be expressed as a convex combination
of at most $d+1$ points in $X$. This implies that every point set that is not in convex position
contains a subset of $d+2$ points that are not in convex position, \ie, a short \emph{certificate}
of non-convexity.  We will further assume that Chan's algorithm for testing of convex position outputs
such a tuple when the input is not in convex position.

\paragraph{The convex position tester of Czumaj, Sohler, and Ziegler.}
The convex position tester of Czumaj~\etal~\cite{CSZ-esa00} draws a random sample of the input set
and makes a decision based on the convexity of this sample.  The algorithm is set up to work in $\RR^d$,
for any fixed dimension $d$. Given $\eps>0$, the tester accepts every point set in convex position,
and rejects every point set that is $\eps$-far from convex position with probability at least $2/3$. 
If the input is not in convex position and is not $\eps$-far from convex position, the outcome of the
algorithm can go either way, \ie, there is no specified action for the situation in-between. 
Most of the technical justification is unpublished; for the present time, it can be found online~\cite{CSZ00}.
The authors present two testers for convex position: \texttt{Convex-A} and \texttt{Convex-B},
see~\cite[p.~161]{CSZ-esa00}:

\bigskip
\noindent{{\bf Algorithm~\texttt{Convex-A}}}
\begin{itemize} \itemsep 1pt
\item[] Step 1: Choose a subset $S \subset P$ of size
$s = 36 \cdot n^{\frac{d}{d+1}} \eps^{-\frac{1}{d+1}}$ uniformly at random.
\item[] Step 2: Compute all $h$ extreme points of $S$.
\item[] Step 3: If $h < n$ then \underline{reject} else \underline{accept}.
\end{itemize}

\bigskip
\noindent{{\bf Algorithm~\texttt{Convex-B}}}
\begin{itemize} \itemsep 1pt
\item[] Step 1: Choose a subset $S \subset P$ of size $s = 4/\eps$ uniformly at random.
\item[] Step 2: For each $p \in S$ [simultaneously] check whether $p$ is extreme for $\conv(P)$.
  If $p$ is not extreme for $\conv(P)$ then \underline{exit loop} and \underline{reject}.
\item[] Step 3: If all checks complete, \underline{accept}.
\end{itemize}

The \emph{query complexity}, \ie, the  number of points requested from an \emph{oracle} to perform
the testing, is  $O(n^{d/(d+1)} \eps^{-1/(d+1)})$, which is claimed by the authors to be optimal
(no proof is provided)~\cite{CSZ-esa00}.
The corresponding running time follows from Corollary~\ref{cor:chan} and is subquadratic
in any fixed dimension~$d$.  

The correctness proof for \texttt{Convex-A} is only sketched in~\cite{CSZ-esa00}.
It is however similar in nature to the revised argument we give here based on
Lemmata~\ref{lem:one}, \ref{lem:many} and~\ref{lem:sample}.
The correctness proof for \texttt{Convex-B}, also omitted in~\cite{CSZ-esa00}, is implied from the following.

\begin{lemma} \label{lem:far} 
Let $P \subset \RR^d$ be $\eps$-far from convex position. Then $|P \setminus \Ext(P)| > \eps |P|$.
\end{lemma}
\begin{proof}
  Assume for contradiction that $|P \setminus \Ext(P)| \leq \eps |P|$.
  Removing all points in $P \setminus \Ext(P)$ yields a convex set and thus
  $P$ is $\eps$-close to convex position, a contradiction.
\end{proof}

In fact the sample size in \texttt{Convex-B} can be reduced in half;
\ie, one can set $s=2/\eps$, see below.
If the input $P \subset \RR^d$ is $\eps$-far from convex position, 
then the set $Q=P \setminus \Ext(P)$ is large enough and the tester would reject $P$
if at least one sample point is in $Q$. Since $|Q| \geq \eps |P|$, we have
\[ \Prob(S \cap Q = \emptyset) \leq (1-\eps)^{2/\eps} \leq \e^{-2} \leq \frac13, \]
by applying the standard inequality $1-x \leq \e^{-x}$ for $0 \leq x \leq 1/2$.
Thus $P$ is rejected with probability at least $2/3$, as required.
Note that an input in convex position is accepted by either tester.
In summary, by Corollary~\ref{cor:chan} and Corollary~\ref{cor:chan2},
negative testing  (via \texttt{Convex-A} or \texttt{Convex-B}) can be
accomplished in time
\begin{equation}
  O\left( \min \left \{ T\left(n^{\frac{d}{d+1}} \eps^{-\frac{1}{d+1}}, n^{\frac{d}{d+1}} \eps^{-\frac{1}{d+1}} \right),
  T\left(n, \eps^{-1} \right) \right \} \right).
\end{equation}

Unfortunately, the convex position tester of Czumaj~\etal~\cite{CSZ-esa00} suffers from structural
and performance issues as explained below. 
One issue is an unreasonable dependence of the tester \texttt{Convex-A} of the input parameter $\eps$;
a second concerns a technical lemma that needs correction.
Moreover, as mentioned earlier, most of the claims made in~\cite{CSZ-esa00} are unverifiable since
most proofs are omitted. Here we fix these problems and obtain a more performant negative tester.
Further, its functionality is expanded by including positive certificates. Our paper is self-contained
with all needed proofs included.

(i)~The sample size used by tester \texttt{Convex-A} is
\[ s = 36 \cdot n^{\frac{d}{d+1}} \eps^{-\frac{1}{d+1}}. \]
Since $s \leq n$ is a prerequisite for using the tester, this imposes the restriction $36^{d+1} \leq \eps n$;
equivalently, $\eps \geq 36^{d+1}/n$. 
Since $\eps<1$, this implies $n > 36^{d+1}$. This requirement makes the tester impractical 
even for moderate values of $d$. For instance, if $d=20$, tester \texttt{Convex-A} can only test sets
with $n > 4.8 \cdot 10^{32}$ points. Similarly, if $d=50$, tester \texttt{Convex-A} can only test sets
with $n > 2.3 \cdot 10^{79}$ points, which is approximately the number of atoms in the observable universe.
Arguably, such applications, if any, are rare. As such, the tester isn't functional in the range $d \geq 50$.
In contrast, our Algorithm \texttt{Convex-} in Subsection~\ref{subsec:negative} is only subject to the
very modest restriction $\eps \geq (d+1)/n$. Similarly, our Algorithm \texttt{Convex+} 
in Subsection~\ref{subsec:positive} is subject to very modest restrictions. 

\old{
  (ii)~Another issue is in regard to the correctness of the tester \texttt{Convex-B} in view of the sample size
$s =4/\eps$ used by the tester.
Suppose that $d=4$ and the input is an $n$-element point set that is $\eps$-far from convex position
for a constant $\eps$, say $\eps=1/4$, but not for a larger $\eps$.
By the optimality of the testing sample $s$ mentioned above,
it is required that $s =\Omega(n^{4/5} \eps^{-1/5})$. For $s =4/\eps$, this implies $\eps =O(1/n)$, which
does not hold for large $n$. The tester \texttt{Convex-B} is therefore incorrect; its output is incorrect
most of the time for the input described above and many others.
}  % old

(ii)~Another issue is the correctness of Lemma 3.4 in~\cite{CSZ00}, discussed in Section~\ref{sec:appendix}.
Our Lemma~\ref{lem:sample} is proposed as a replacement.

\paragraph{Our results.}
  We revisit the problem of property testing for convex position for point sets in $\RR^d$.
  Our results draw from previous design and ideas of Czumaj, Sohler, and Ziegler (ESA 2000).
  First, the algorithm is redesigned and its analysis is revised for correctness.
  Second, its functionality is expanded by (i)~exhibiting both negative and positive certificates
  along with the convexity determination, and (ii)~significantly extending the input range
  for moderate and higher dimensions. 
The tester is implemented by two procedures: \texttt{Convex-} and \texttt{Convex+}.
Both run in $O\left(n^{\frac{2\lfloor d/2 \rfloor}{\lfloor d/2 \rfloor +1}} \log^{O(1)}{n}\right) = o(n^2)$ time,
for every $n$ and $\eps$.

\smallskip
The behavior of Algorithm \texttt{Convex-} can be summarized as follows. Let $0 <\eps <1$ be an input parameter. 
\begin{enumerate} \itemsep 1pt
\item If $P$ is in convex position, the algorithm accepts $P$.
\item If $P$ is $\eps$-far from convex position, with probability at least $2/3$ the algorithm rejects $P$ and 
  outputs a $(d+2)$-point witness of non-convexity (as a negative certificate). 
\end{enumerate}

The behavior of Algorithm \texttt{Convex+} can be summarized as follows. Let $0 <\eps <1$ be an input parameter,
and $0< \delta \leq 1/2$ be an adjustable parameter.  
\begin{enumerate} \itemsep 1pt
\item If $P$ is in convex position, the algorithm accepts $P$.
\item If $P$ is $\eps$-close  to convex position for some $\eps>0$ that satisfies $n^{-1} \leq \eps \leq n^{\delta-1}$,
  with probability at least $2/3$ the algorithm accepts $P$ and outputs a $1/(6 n^\delta)$-approximation of the
  largest subset in convex position  as a positive certificate.
\end{enumerate}

\paragraph{Related work.}
Two early articles in the area of property testing are due to Blum~\etal~\cite{BLR93} and
Erg{\"u}n~\etal~\cite{EKK+00}.
Besides testing for convex position, testing for other geometric properties has been considered
in~\cite{CSZ-esa00}: pairwise disjointness of a set of generic bodies, disjointness of two polytopes,
and Euclidean minimum spanning tree verification.  
A~continuation of the work in~\cite{CSZ-esa00} appears in~\cite{CS-esa01}.
A more recent article on property testing for point sets in the plane is due to Han~\etal~\cite{HKSS18}.
Two recent monographs dedicated to the general subject of property testing are~\cite{BY22} and~\cite{Gol17}.
The topic of property testing, including testing for convex position, is also addressed in a recent book by
Eppstein~\cite{Epp18}. A question from that book is discussed in Section~\ref{sec:remarks}.

\section{An enhanced functionality tester for convex position}  \label{sec:main}

The tester is implemented by two procedures: Algorithm \texttt{Convex-}  (in Subsection~\ref{subsec:negative})
and Algorithm \texttt{Convex+} (in Subsection~\ref{subsec:positive}).
The two procedures may be run independently of each other. 
The goal of Algorithm \texttt{Convex-} is rejecting point sets that are far from convex position;
whereas that of Algorithm \texttt{Convex+} is accepting point sets that are close to convex position. 
Each algorithm exhibits a suitable certificate along with its probabilistic determination.
While the decision is randomized, the certificates produced are indisputable, \ie,
a negative certificate is always a $(d+2)$-point set that is not in convex position,
and a positive certificate output by Algorithm \texttt{Convex+} is always a $1/(6 n^\delta)$-approximation of the
largest subset in convex position.

\paragraph{Common tools.}
A randomized algorithm for generating a random $s$-set for a given $s$, $1 \leq s \leq n$, 
in $O(s \log{s})$ time (and $O(s)$ expected time) from~\cite[Ch.~4]{NW78}, can be used to implement
random sample selection.  Alternatively, a linear-time algorithm for the same task from~\cite[Sec~5.2]{RND77}
can also be used.

\subsection{Negative testing: Algorithm \texttt{Convex-}} \label{subsec:negative}

Several constraints among the input parameters need to be respected usually for technical reasons.
In particular, it is assumed that (note that these constraints are very mild):
\begin{itemize} \itemsep 1pt
\item $n \geq 2^{10}$,  this is needed in the proof of Lemma~\ref{lem:sample}.
\item $n \geq 32(d+1)$, this ensures that $\ell \leq n/32$ when using Lemma~\ref{lem:sample}.
\item $\eps \geq \frac{10(d+1)}{n}$, this ensures that $k \geq 10$ in Step~1;
  compare this to the constraint $\eps \geq 36^{d+1}/n$ in tester \texttt{Convex-A} that restricts its
  use to low dimensions. 
 \item  $\eps \leq \frac{d-1}{2d}$,  this ensures $\frac{(1-\eps)}{d+1} \geq \frac{1}{2d}$ in the analysis. 
\end{itemize}

\smallskip
\noindent{{\bf Algorithm~\texttt{Convex-}}}
\begin{itemize} \itemsep 1pt
\item[] Step 1:  Let $k= \lfloor \frac{\eps n}{d+1} \rfloor$, $\ell=d+1$, 
  $s_0= \ell + \frac{n-\ell}{(2k)^{1/\ell}}$, and $s = \lceil s_0 \rceil$.
Repeat Step~2 and Step~3 in succession up to $22$ times.
\item[] Step 2:  Randomly select a subset $S \subset P$ of size $s$,
  with all $s$-subsets being equally likely. 
\item[] Step 3: Test $S$ for convex position using Chan's algorithm. 
  If $S$ is not in convex position, output a $(d+2)$-point witness of non-convexity and reject $P$.
  Otherwise go to Step~2 for the next repetition.
\item[] Step 4: If all $22$ samples were determined to be in convex position, accept $P$. 
\end{itemize}

\paragraph{Time analysis.}
It is easily verified that the setting for $s$ in Step~1 yields
\[ s = \Theta \left( n^{\frac{d}{d+1}} \eps^{-\frac{1}{d+1}} \right). \]
This is in accordance with the choice of the sample size for Algorithm \texttt{Convex-A} in~\cite{CSZ-esa00}. 
As such, the runtime of Algorithm~\texttt{Convex-} is
\begin{align*}
T(s,s) &= O \left( T  \left( n^{\frac{d}{d+1}} \eps^{-\frac{1}{d+1}}, n^{\frac{d}{d+1}} \eps^{-\frac{1}{d+1}} \right) \right) \\
&= O\left(n^{\frac{d}{d+1} \cdot \frac{2\lfloor d/2 \rfloor}{\lfloor d/2 \rfloor +1}}
\cdot \eps^{-\frac{1}{d+1} \cdot \frac{2\lfloor d/2 \rfloor}{\lfloor d/2 \rfloor +1}}
\cdot \log^{O(1)}{(n/\eps)} \right).
\end{align*}
Since $\eps =\Omega(1/n)$, the above expression becomes
\[ T(s,s) = O\left( T(n,n) \right) = O\left(n^{\frac{2\lfloor d/2 \rfloor}{\lfloor d/2 \rfloor +1}} \log^{O(1)}{n}\right) = o(n^2),
\text{ for every }n \text{ and } \eps. \]
This can be also seen directly: since $s \leq n$, $T(s,s) \leq T(n,n)= o(n^2)$. 
%for every $n$ and $\eps$. This can be also seen directly: since $s \leq n$, $T(s,s) \leq T(n,n)= o(n^2)$. 

\paragraph{Rejecting the input with probability $\geq 2/3$.}
Assume that $P$ is $\eps$-far from convex position.
We show that with probability at least $2/3$, Algorithm~\texttt{Convex-}
rejects the input in step 3 and outputs a suitable $(d+2)$-point witness. 
We first recall the following lemmas (analogous to Lemma 3.1 and 3.2 from~\cite{CSZ00}),
slightly rewritten here for convenience.
%Its proof can be found in Section~\ref{sec:appendix}. 

\begin{lemma} \label{lem:one} {\rm (An earlier version in~\cite{CSZ00})}. 
Let $P \subset \RR^d$ be a set of $n$ points that is not in convex position and $p \in P$ be an interior point.
Then there exist points $p_1,\ldots,p_d \in P$ and $U \subset P \setminus \{p_1,\ldots,p_d,p\}$ with
$|U| \geq \frac{n-1}{d+1}$ such that $\{p_1,\ldots,p_d,p\} \cup \{q\}$ is not in convex position for every $q \in U$;
more precisely, $p$ is an interior point in the simplex $\Delta(p_1,\ldots,p_d,q)$ for every $q \in U$.
\end{lemma}
\begin{proof}
  Since $p \in P$ is an interior point, by Carath\'eodory's Theorem and by the general position
  assumption, there exists a set $W \subset P$ of size $d+1$ such that $p \in \mathring{W}$.
  See Fig.~\ref{fig:cone}.

\begin{figure}[htbp]
\centering
\includegraphics[scale=0.9]{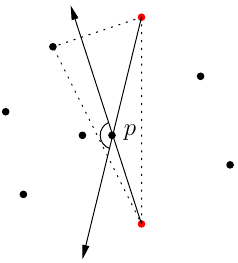}
\caption{$P$ is a set of $9$ points in the plane. The cone determined by the two red points
  contains $4 \geq 8/3$ points in $P$.}
\label{fig:cone}
\end{figure}

  Denote by $W_i$, $i=1,\ldots,d+1$, the $d+1$ subsets of $W$ of size $d$.  We show that
  one of the subsets $W_i$ of $W$ satisfies the requirement in the lemma. We may assume without loss of generality
  that $p=(0,\ldots,0)$. We partition $\RR^d$ into $d+1$ cones as follows. Let $W^{-}_i$, $i=1,\ldots,d+1$,
  denote the set of points $\{(-x_1,\ldots,-x_d) \colon (x_1,\ldots,x_d) \in W_i\}$. 
  The conic combination of the point vectors in the set $W^{-}_i$ defines a cone $C_i$, $i=1,\ldots,d+1$. 
  The union of these cones cover $\RR^d$. Thus there is a cone $C_j$, $1 \leq j \leq d+1$, that contains
  at least  $\frac{n-1}{d+1}$ points in $P$. Observe that for every $q \in P \cap C_j$ we have
  $ p \in \mathring{(W_j \cup \{q\})}$. Consequently, one can set $\{p_1,\ldots,p_d\} = W_j$
  to conclude the proof. 
\end{proof}

The following lemma applies to point sets that are far from convex position. The sets $W_i$ and $U_i$ constructed
in the lemma are fixed before the samplings and are only used in the algorithm analysis.

\begin{lemma} \label{lem:many}  {\rm (An earlier version in~\cite{CSZ00})}. 
  Let $P \subset \RR^d$ be a set of $n$ points that is $\eps$-far from convex position
  and let $k= \lfloor \frac{\eps n}{d+1} \rfloor$.
  Then there exist sets $W_i, U_i \subset P$ for $1 \leq i \leq k$, such that the following conditions are satisfied:
\begin{enumerate}  [{\rm (i)}] \itemsep 1pt
\item $|W_i|=d+1$ for $1 \leq i \leq k$,
\item $W_i \cap W_j= \emptyset$ for all $1 \leq i<j \leq k$,
\item $W_i \cap U_i = \emptyset$ for $1 \leq i \leq k$,
\item $W_i \cup \{q\}$ is not in convex position for every $q \in U_i$, and
\item $|U_i| \geq \frac{n}{d+1} -k$  for $1 \leq i \leq k$. In particular,
 $|U_i| \geq \frac{(1-\eps) n}{d+1}$. 
\end{enumerate}
\end{lemma}
\begin{proof}
  We construct point sets $P_1,P_2,\ldots,P_k$ iteratively. We initially set $P_1:=P$ and then
  iteratively find $W_i \subset P_i$ and set $P_{i+1}:= P_i \setminus W_i$ for $i=1,\ldots,k$.
  By construction the sets $W_i$ are pairwise disjoint, as required. 
  Assuming that $|W_i|=d+1$ for $1 \leq i \leq k$, implies that
\[ |P_i|= n - (d+1)(i-1) \geq n - (d+1)(k-1) > n - (d+1) \frac{\eps n}{d+1} = (1-\eps) n. \] 
By the assumption in the lemma, $P_i$ cannot be in convex position. By Lemma~\ref{lem:one} there exist
$p_1,\ldots,p_d,p \in P_i$ and $U_i \subset P_i \setminus \{p_1,\ldots,p_d,p\}$ with
\begin{align*}
|U_i| &\geq \frac{|P_i|-1}{d+1} \geq \frac{n -(d+1)(i-1)-1}{d+1} \geq \frac{n -(d+1)(k-1)-1}{d+1}\\
&> \frac{n}{d+1} -k =  \frac{n}{d+1} - \frac{\eps n}{d+1} = \frac{(1-\eps) n}{d+1},
\end{align*}
such that $p$ is an interior point in the simplex $\Delta{p_1,\ldots,p_d,q}$ for every $q \in U_i$.
Let $W_i:= \{p_1,\ldots,p_d,p\}$ and observe that $W_i \cap U_i = \emptyset$.
Note that all properties in the lemma have been verified.
\end{proof}

We also need another lemma suggested by Czumaj~\etal~\cite{CSZ00}.
Here we include a proof that follows the ideas of the original proof, however, it is revised for correctness
and for a slightly restricted range of the parameters that suffices for our purposes.
More details can be found in Section~\ref{sec:appendix}.

\begin{lemma} \label{lem:sample} {\rm (An earlier version in~\cite{CSZ00})}. 
  Let $\Omega$ be a set of size $n$ and $W_1,W_2,\ldots,W_k$ $ \subset \Omega$ be $k$ pairwise disjoint subsets
  of $\Omega$ of size $\ell$, where $k \geq 10$ and $3 \leq \ell \leq n/32$.
  Let $s$ be a positive integer such that   $\ell + \frac{n-\ell}{(2k)^{1/\ell}} \leq s \leq n$ 
  and $S \subset \Omega$ be a subset of $\Omega$ of size $s$ chosen uniformly at random. Then
  \[ \Prob(\exists i \leq k \colon (W_i \subset S)) \geq \frac14. \]
\end{lemma}
\begin{proof}
  Observe that $k \ell \leq n$, hence $k \leq n/\ell$.
  Let $s_0$ be the real number defined as follows:
\begin{equation} \label{eq:s}
  s_0= \ell + \frac{n-\ell}{(2k)^{1/\ell}}, \text{ or } k \left(\frac{s_0-\ell}{n-\ell} \right)^\ell =\frac12,
\end{equation}
and note that $\ell < s_0 < n$. Indeed, the lower bound is clear and the upper bound $s_0<n $
is equivalent to $(2k)^{1/\ell} >1$ which is obvious. We first prove that
\begin{equation} \label{eq:6ell}
  s_0 \geq 3 \ell \log{k}.
\end{equation}
It suffices to show that $n-\ell \geq 3 \ell (2k)^{1/\ell} \log{k}$, or, since $\ell \leq n/32$, that
% xxx: recent change to reduce space
  $3 \ell (2k)^{1/\ell} \log{k} \leq \frac{31n}{32}$.
  We have
  \[ 3 \ell (2k)^{1/\ell} \log{k} \leq 3 \ell \left( \frac{2n}{\ell} \right)^{1/\ell} \log{\left( \frac{2n}{\ell} \right)}
    \leq  \frac{31n}{32}.  \]
  Indeed, a standard verification shows that the function
  \[ f(x) = 3 x \left ( \frac{2n}{x} \right)^{1/x} \log{\left( \frac{2n}{x} \right)}, x \in \left[3,\frac{n}{32}\right], \]
  where $n \geq 2^{10}$, attains it maximum at $x= n/32$, thus
\begin{align*}
f(x) &\leq f\left(\frac{n}{32}\right) =
  3 \cdot \frac{n}{32} \cdot \left( \frac{2n}{n/32}\right)^{32/n} \log{ \left(\frac{2n}{n/32} \right)} \\
  &= \frac{3n}{32} \cdot 64^{32/n} \cdot \log{64} \leq \frac{18n}{32} \cdot \frac54 \leq \frac{31n}{32}.
\end{align*}
This concludes the proof of~\eqref{eq:6ell} and we next focus on the inequality in the lemma. 

\smallskip
Since the probability in question increases as the sample size $s$ grows, it suffices
to prove the inequality for $s = \lceil s_0 \rceil$. Observe that $\ell +1 \leq s \leq n$.  
By the Boole-Bonferoni inequality---see, \eg, ~\cite[Ch.~2]{Lov93}, we have
\begin{equation} \label{eq:bb}
  \Prob(\exists i \leq k \colon (W_i \subset S)) \geq \sum_{i=1}^k \Prob(W_i \subset S) -
 \sum_{1 \leq i<j \leq k} \Prob((W_i \cup W_j) \subset S). 
\end{equation}

It is easily verified that
\begin{align*}
  \Prob(W_i \subset S) &= \frac{{n-\ell \choose s-\ell}}{{n \choose s}}
  =\frac{(n-\ell)!}{(s-\ell)! (n-s)!} \cdot \frac{s! (n-s)!}{n!} \\
  &= \frac{(n-\ell)! s!}{n! (s-\ell)!} 
  = \prod_{r=0}^{\ell-1} \frac{s-r}{n-r}, \text{    and } \\
  \Prob((W_i \cup W_j) \subset S) &= \frac{{n-2\ell \choose s-2\ell}}{{n \choose s}}
  = \prod_{r=0}^{2\ell-1} \frac{s-r}{n-r}\\
&= \prod_{r=0}^{\ell-1} \frac{s-r}{n-r} \cdot \prod_{r=0}^{\ell-1} \frac{(s-\ell)-r}{(n-\ell)-r},  \text{    for } 1\leq i<j \leq k.
\end{align*}

Substituting these into Inequality~\eqref{eq:bb} and finally using~\eqref{eq:s} yields
\begin{align*}
\Prob(\exists i \leq k \colon (W_i \subset S)) &\geq k  \cdot \prod_{r=0}^{\ell-1} \frac{s-r}{n-r}
 - {k \choose 2} \cdot \prod_{r=0}^{\ell-1} \frac{s-r}{n-r} \cdot \prod_{r=0}^{\ell-1} \frac{(s-\ell)-r}{(n-\ell)-r} \\
 &= k \cdot \prod_{r=0}^{\ell-1} \frac{s-r}{n-r}
 \left( 1 - \frac{k-1}{2} \cdot   \prod_{r=0}^{\ell-1} \frac{(s-\ell)-r}{(n-\ell)-r} \right) \\
 &\geq k \cdot \prod_{r=0}^{\ell-1} \frac{s-\ell}{n-\ell} \cdot 
 \left( 1 - \frac{k}{2} \cdot   \prod_{r=0}^{\ell-1} \frac{s-\ell}{n-\ell} \right)\\
&=  k \cdot \left( \frac{s-\ell}{n-\ell} \right)^\ell \cdot  \left( 1 - \frac{k}{2} \cdot  \left( \frac{s-\ell}{n-\ell} \right)^\ell \right).
\end{align*}

Let
\[ F_1 =  k \cdot \left( \frac{s-\ell}{n-\ell} \right)^\ell  \text{ and }
F_2 = 1 - \frac{k}{2} \cdot  \left( \frac{s-\ell}{n-\ell} \right)^\ell . \]
It suffices to show that $F_1 \geq \frac12$ and $F_2 \geq \frac12$.
For the first inequality, we have 
\begin{align}    \label{eq:factor1}
  F_1 &= k \cdot \left( \frac{s-\ell}{n-\ell} \right)^\ell  \geq 
  k \cdot \left( \frac{s_0-\ell}{n-\ell} \right)^\ell  = \frac12.
\end{align}

For the second, recall that $0 \leq s-s_0 <1$ and $s_0 \geq 6 \ell \geq 3 \ell$ by~\eqref{eq:6ell}.
Applying the standard inequality $1+x \leq \e^x$ for $0 \leq x \leq 1/2$ yields:
\begin{equation}   \label{eq:exp}
  \left( \frac{s-\ell}{s_0-\ell} \right)^\ell = \left( 1 + \frac{s-s_0}{s_0-\ell} \right)^\ell
  \leq \left( 1 + \frac{1}{2 \ell} \right)^\ell \leq \exp(0.5) \leq 2. 
\end{equation}

Using~\eqref{eq:exp} and~\eqref{eq:s} once again yields
\begin{align} \label{eq:factor2}
 F_2 &= 1 - \frac{k}{2} \cdot  \left( \frac{s-\ell}{n-\ell} \right)^\ell = 
1- \left( \frac{s-\ell}{s_0-\ell} \right)^\ell \cdot  \frac{k}{2} \cdot \left( \frac{s_0 -\ell}{n-\ell} \right)^\ell \nonumber \\
&\geq 1 - 2 \cdot \frac{k}{2} \cdot \left( \frac{s_0 -\ell}{n-\ell} \right)^\ell =
1- k \cdot  \left( \frac{s_0 -\ell}{n-\ell} \right)^\ell  =  \frac12. 
\end{align}

Consequently, we have
\[ \Prob(\exists i \leq k \colon (W_i \subset S)) \geq F_1 \cdot F_2 \geq \frac12 \cdot \frac12 = \frac14, \]
as required.
\end{proof}

Let $k= \lfloor \frac{\eps n}{d+1} \rfloor$, $\ell=d+1$, and recall that Algorithm \texttt{Convex-}
sets $s =\lceil s_0 \rceil$, where $s_0$ is given by Equation~\eqref{eq:s}.

We next prove that the algorithm finds the sample $S$ not convex with probability $\geq 1/20$ in each of the 
$22$ repetitions in Step 2 and Step 3. 
Consider one execution of Step 2 and Step 3. 
For a fixed $i \leq k$, let $F_i$ be the event that $S \cap U_i = \emptyset$.
By Lemma~\ref{lem:many}, we have $|U_i| \geq \frac{(1-\eps)n}{d+1} \geq \frac{n}{2d}$.
Observe that
\[ \left(1-\frac{1}{2d}\right)^{d+1} \leq \frac23, \text{ for } d \geq 2. \]
By~\eqref{eq:6ell} we have $s \geq s_0 \geq 3 \ell \log{k}$, thus
(recall also that $k \geq 10$, which us used in the last inequality of the chain below)
\begin{align*}
\Prob(F_i) &= \Prob(S \cap U_i = \emptyset) =
\frac{{n -|U_i| \choose s}}{{n \choose s}} \\
&= \frac{(n-|U_i|)(n-|U_i|-1) \cdots (n-|U_i| -s+1)}{n(n-1) \cdots (n-s+1)} \leq
\left( 1 - \frac{|U_i|}{n} \right)^s \\
&\leq \left(1-\frac{1}{2d}\right)^s \leq \left(1-\frac{1}{2d}\right)^{3 \ell \log{k}} \\
&\leq \left( \frac23 \right) ^{3 \log{k}} \leq \frac{1}{5k}, \text{ for } i \in [k] \text{ and } d \geq 2.
\end{align*}

Let $E_1$ be the event that $S \cap U_i \neq \emptyset$ for every $i \leq k$.
By the union bound, we deduce that
\[ \Prob(\overline{E_1}) \leq k \cdot \Prob(F_1) \leq \frac15. \]

Let $E_2$ be the event that there exists $i \leq k$ such that $W_i \subset S$.
We next verify that the inequality $\ell + \frac{n-\ell}{(2k)^{1/\ell}} \leq s \leq n$ 
specified in Lemma~\ref{lem:sample} holds. 
Indeed,
\[ s = \lceil s_0 \rceil \geq s_0= \ell + \frac{n-\ell}{(2k)^{1/\ell}}, \]
and $s_0<n$ as shown in the proof of Lemma~\ref{lem:sample}, whence $s = \lceil s_0 \rceil \leq n$.
Hence by Lemma~\ref{lem:sample} we have 
  \[ \Prob(E_2) = \Prob(\exists i \leq k \colon (W_i \subset S)) \geq \frac14. \]

Putting these bounds together yields
\begin{align*}
\Prob(E_1 \cap E_2) &=1 - \Prob(\overline{E_1} \cup \overline{E_2}) \geq
1 -\Prob(\overline{E_1}) - \Prob(\overline{E_2}) \\
&\geq 1 -  \frac{1}{5} - (1- \Prob(E_2)) = \Prob(E_2) - \frac{1}{5}\\
& \geq \frac14 - \frac{1}{5} = \frac{1}{20}.
\end{align*}

Let $E$ be the event that Algorithm \texttt{Convex-} finds the sample not convex in at least one
of the $22$ executions of Step 2 and Step 3. The $22$ repetitions are independent events, thus
\[ \Prob(E) \geq 1 - \left(1- \frac{1}{20} \right)^{22} \geq \frac23. \]
Thus with probability at least $2/3$, Algorithm \texttt{Convex-} rejects the input, as required.

\subsection{Positive testing: Algorithm \texttt{Convex+}} \label{subsec:positive}

Assume for technical reasons that $n$ is sufficiently large: $n \geq 1500$.
Let $0< \delta \leq 1/2$ be an adjustable parameter. 
Assume that $P$ is $\eps$-close to convex position
for some $\eps>0$, where $n^{-1} \leq \eps \leq n^{\delta-1}$; note, this means
that $P$ can be made convex by removing at most $\eps n \leq n^\delta$ points. 

\smallskip
\noindent{{\bf Algorithm~\texttt{Convex+}}}
\begin{itemize} \itemsep 1pt
\item[] Step 1: Randomly select a subset $S \subset P$ of size $s= \lceil 1/(6 \eps) \rceil $,
  with all $s$-subsets being equally likely. 
\item[] Step 2: Test $S$ for convex position using Chan's algorithm. 
  If $S$ is not in convex position, output a $(d+2)$-point witness of non-convexity and reject $P$.
  Otherwise  output $S$ as a subset in convex position and accept $P$. 
\end{itemize}

\paragraph{Time analysis.}
The setting $s= \lceil 1/(6 \eps) \rceil $ in Step~1 yields
that the runtime of Algorithm~\texttt{Convex+} is
\begin{align*}
T(s,s) &= O \left( T( 1/\eps, 1/\eps) \right) 
= O\left(\eps^{-\frac{2\lfloor d/2 \rfloor}{\lfloor d/2 \rfloor +1}} \log^{O(1)}{1/\eps}\right).
\end{align*}
Since $\eps = \Omega(1/n)$, 
\[ T(s,s) = O\left( T(n,n) \right) =
O\left(n^{\frac{2\lfloor d/2 \rfloor}{\lfloor d/2 \rfloor +1}} \log^{O(1)}{n}\right) = o(n^2), 
\text{ for every }n \text{ and }\eps.   \]

\paragraph{Accepting the input with probability $\geq 2/3$.}
We next show that with probability at least $2/3$, Algorithm \texttt{Convex+}
accepts $P$ and outputs a subset of size $\lceil 1/(6 \eps) \rceil$ of $P$ in convex position.
By the assumption we can write $P= C \, \cup \, D$, where $C$ is in convex position and
$|D| \leq \eps n =:t$. Recall that $s=\lceil 1/(6 \eps) \rceil$. 
Note that
\[ st = \left \lceil \frac{1}{6 \eps} \right \rceil \cdot \eps n \leq \frac{1}{6 \eps} \cdot \eps n +\eps n
= \frac{n}{6} + \eps n \leq \frac{100n}{595} \text{ for } n \geq 1500. \]
Indeed, $ n \geq 1500 \implies n^{0.9} \geq 721 \implies \eps \leq 1/n^{0.9} \leq 1/721$,
for which the above inequality holds. In particular, we have $t \leq st \leq 100n/595$. 
We show that
\[   \Prob(S \cap D =\emptyset) = \Prob(S \subseteq C) \geq \frac23. \]
Applying the standard inequality $1-x \geq \e^{-2x}$ for $0 \leq x \leq 1/2$ yields:
\begin{align*}
  \Prob(S \subseteq C) &= \frac{{|C| \choose s}}{{n \choose s}} \geq \frac{{n-t \choose s}}{{n \choose s}} 
  = \frac{(n-s)(n-s-1) \cdots (n-s -t+1)}{n(n-1) \cdots (n-t+1)}\\
  &= \prod_{i=0}^{t-1} \left( 1 - \frac{s}{n-i} \right) 
  \geq \left(1 - \frac{s}{n-t+1} \right)^t \geq \exp \left(\frac{-2st}{n-t+1} \right)\\
  &\geq \exp \left(\frac{-200}{495} \right) \geq \frac23, 
\end{align*}
as required. Hence with probability at least $2/3$, $S$ is determined to be in convex position and output
by the algorithm, as required. Let $\opt$ denote the size of  the largest convex subset of $P$. 
Since $\opt \leq n$ and $\eps n \leq n^\delta$, the approximation ratio of
Algorithm \texttt{Convex+} is 
\[ \frac{s}{\opt} \geq \frac{s}{n} = \left \lceil \frac{1}{ 6 \eps } \right \rceil \frac{1}{n} \geq 
\frac{1}{ 6 \eps n} \geq \frac{1}{6 n^\delta}. \]
In particular, when $\delta=0.1$, the ratio is at least $1/24$ for all $n \leq 10^6$.

\section{Concluding remarks}   \label{sec:remarks}

\paragraph{Summary.}
We presented and analyzed a convexity-testing algorithm implemented by two procedures based on random sampling
that has the following enhanced functionality:

\begin{enumerate}  \itemsep 1pt
\item For point sets that are $\eps$-far from convex position,  with probability $\geq 2/3$ the algorithm outputs 
 a $(d+2)$-point witness of non-convexity as a negative certificate.
\item For point sets that are $\eps$-close  to convex position, with probability $\geq 2/3$ 
  the algorithm outputs a $1/(6 n^\delta)$-approximation of a maximum-size convex subset.
  [Comment: The current fastest algorithm for computing the largest subset in convex position
  takes $O(n^3)$ time for $d=2$, see~\cite{CK80,EG89}.
In contrast, the problem of computing a largest subset of points in convex position
is $\NP$-complete for $d \geq 3$~\cite{GKW13}, and moreover, no approximation
algorithm is known.]
\item The input range for the tester is significantly extended --- for moderate and higher dimensions ---
  compared to the previous version in~\cite{CSZ-esa00}. 
\end{enumerate}

\paragraph{A clarifying remark (A question of Eppstein for the planar case).}
Four-point witnesses to non-convexity can be also viewed as forbidden configurations or obstacles in
a convex set of points. Taking this view, sample-based property testing attains the following performance when
the sample size is chosen based on the structure of the obstacle set.

\begin{theorem} {\rm \cite[Theorem 6.8]{Epp18}}  \label{thm:general}
  Let $O_1,O_2,\ldots$ be  a finite set of obstacles, whose maximum size is $t$, and let $\eps$ and $p$ be numbers
  in the range $0 <\eps<1$ and $0<p<1$. Then there is a sample-based property testing algorithm for the
  property that avoids these obstacles whose sample size, on configurations of size $n$, is $O(n^{1-1/t})$ and
  whose false positive rate for configurations that are $\eps$-far from this property is at most $p$.
\end{theorem}

Recall that a \emph{sawtooth} configuration of $n$ points (where $n$ is a multiple of $4$)
is obtained by adding $n/2$ points very close to the midpoints of the $n/2$ sides of a regular $n/2$-gon
and interior to it~\cite[Definition 3.9]{Epp18}. It is known that a sawtooth configuration of $n$ points is $1/4$-close to
convex \ie, it can be made convex by removing a quarter, but not fewer, of its points; see, \eg, \cite[Observation~1.11]{Epp18}.
By Theorem~\ref{thm:general}, letting $t=4$ (by the witness structure), $\eps=1/4$, and $p=1/3$, 
indicates that a sample-based convexity testing algorithm with sample size $O(n^{3/4})$ achieves a false positive
rate at most $1/3$ for configurations that are $1/4$-far from convexity. 

Likely unaware of the work of Czumaj~\etal~\cite{CSZ-esa00,CS-esa01}, Eppstein asked the following natural
question~\cite[Open Problem 11.10]{Epp18}:
``Does the sample-based property testing algorithm for convexity, with sample size $O(n^{2/3})$,
achieve constant false positive rate, or is sample size $\Omega(n^{3/4})$ needed?''
Here achieving constant false positive rate means assuring that the false positive rate is bounded from above
by a constant. The machinery developed by Czumaj~\etal~for convexity testing
(this includes Lemmas 3.2 and 4.9 in~\cite{CSZ00}) and revisited here in Section~\ref{sec:main}
shows that a sample size $O(n^{2/3})$ suffices for that purpose and in general for any constant
$0 <\eps<1$ and $0<p<1$. This answers Eppstein's question.

\appendix

\section{Remarks on Lemma~3.4 in~\cite{CSZ00}}  \label{sec:appendix}

The following lemma is suggested in~\cite{CSZ00}.
Here we argue why the lemma cannot be used as is.

\begin{lemma} \label{lem:sample-old} {\rm \cite{CSZ00}}. 
  Let $\Omega$ be an arbitrary set set of $n$ elements. Let $k$ and $\ell$ be arbitrary integers
  (possibly dependent on $n$) and let $s$ be an arbitrary integer such that $s \geq 2n/(2k)^{1/\ell}$.
  Let $W_1,W_2,\ldots,W_k$ be arbitrary disjoint subsets of $\Omega$ each of size $\ell$.
  Let $W$ be a subset of $\Omega$ of size $s$ which is chosen independently and uniformly at random. 
  Then
  \[ \Prob(\exists j \in [k] \colon (W_j \subseteq W)) \geq \frac14. \]
\end{lemma}

We make two points:

\smallskip
(i) The first point is minor: taking $s$ as the smallest integer satisfying $s \geq 2n/(2k)^{1/\ell}$,
namely $s = \lceil 2n/(2k)^{1/\ell} \rceil$ may result in an integer larger than $n$ and thereby be infeasible.
For example, the setting $n=256$, $k=8$, $\ell=8$, yields $s = \lceil 2n/(2k)^{1/\ell} \rceil = 363>256$.

\smallskip
(ii) The second point requires attention.
%(ii) The second point is more serious.
Reading through the first few lines of the proof suggests that one could take
\begin{equation} \label{eq:s-old}
  s = \ell + \frac{n-\ell}{(2k)^{1/\ell}}, \text{ or } k \left(\frac{s-\ell}{n-\ell} \right)^\ell =\frac12.
\end{equation}

However, this value may be not an integer, and thereby be again infeasible. Suppose that
one takes instead the ceiling in the expression of $s$:
\begin{equation} \label{eq:s-old-ceil}
  s = \ell + \left \lceil \frac{n-\ell}{(2k)^{1/\ell}} \right \rceil. 
\end{equation}
For the above setting in (i), this yields
$s = 8 + \left \lceil \frac{248}{(16)^{1/8}} \right \rceil = 8 + 176 = 184$. 
Then the two factors that appear in the calculation of the lower bound on the probability
in question are

\begin{align*}
F_1 &=  k \cdot \left( \frac{s-\ell}{n-\ell} \right)^\ell  = 8 \cdot \left( \frac{176}{248} \right)^8 = 0.5147\ldots,\\
F_2 &= 1 - k \cdot  \left( \frac{s-\ell}{n-\ell} \right)^\ell  = 1 - 8 \cdot \left( \frac{176}{248} \right)^8 =0.4852\ldots
\end{align*}

It is now clear that $F_1 \cdot F_2 < \frac14$. Taking the floor does not work either.
The above example is not an exception, and this occurs whenever the value of $s$
in~\eqref{eq:s-old} is not an integer, which happens most of the time.

\end{document}